\documentclass[a4paper,runningheads]{llncs}
\usepackage[T1]{fontenc}
\usepackage{graphicx}
\usepackage{amssymb}
\usepackage{algorithm}
\usepackage{algpseudocode}
\usepackage{amsmath}
\usepackage{float}
\usepackage{placeins}
\usepackage{url}
\usepackage{bbding}

\usepackage[hidelinks]{hyperref}
\usepackage{orcidlink}
\begin{document}
\sloppy
\raggedbottom
\makeatletter
\let\theorem\relax
\let\endtheorem\relax
\makeatother
\spnewtheorem{theorem}{Theorem}{\bfseries}{\itshape}
\title{Behavioral Privacy Leakage in Agentic Negotiation: Formalizing and Mitigating Inference Attacks via Randomized Policies}
\titlerunning{Behavioral Privacy Leakage in Agentic Negotiation}
%
\author{Barkha Rani\textsuperscript{\Envelope}\,\orcidlink{0009-0008-8222-5970}}

\authorrunning{B. Rani}

\institute{Apple Inc.\\
\email{barkha\_rani@apple.com}}

\maketitle              
\begin{abstract}
Autonomous negotiation agents are increasingly deployed in high-stakes settings such as insurance and procurement. While cryptographic techniques protect explicitly disclosed constraint values, they fail to address a subtler threat: behavioral privacy leakage, where an adversary infers private constraints from observable negotiation dynamics such as concession trajectories, timing, and convergence patterns.
This paper investigates behavioral differential privacy in multi-round negotiation protocols. We design an adaptive stochastic negotiation policy that jointly guarantees $(\varepsilon, \delta)$-differential privacy, almost-sure convergence of the offer sequence (reaching agreement when the counterparty's reservation value permits), and high negotiation utility. Evaluated on 3,000 synthetic bilateral negotiations, our mechanism reduces adversarial inference accuracy by 43--50\% while maintaining a negotiation success rate and utility above 90\%, demonstrating that strong privacy guarantees can be achieved without significant loss of performance.
\end{abstract}
\section{Introduction} 

Large language models now power autonomous agents that operate in
high-stakes negotiation settings, including insurance pricing,
procurement contracts, and financial services. These systems act on
behalf of users whose private constraints such as maximum budgets
or reservation values must remain confidential throughout the
negotiation process. Existing defenses leverage cryptographic primitives — including computation sharing protocols, proof-of-knowledge schemes, and lattice-based encryption — to prevent direct exposure of constraint values to opposing parties.

However, cryptographic protection of explicit data does not address a
more subtle threat: the negotiation \emph{behavior} itself constitutes
a side channel. An agent's offer sequence, concession trajectory,
response timing, and convergence speed are all observable by the
counterparty, and together they form a rich behavioral trace from
which private constraints can be inferred even when the underlying
data is cryptographically protected.

Consider a concrete example. Alice employs an autonomous agent to
negotiate a health insurance premium, with a private budget of
\$3,000. The agent is configured with a zero-knowledge proof
preventing the insurer from learning her budget directly.
Nevertheless, the agent opens at \$2,600, advances to \$2,850 in the
second round, and reaches \$2,950 in the third. The accelerating
concession pattern - large early moves tapering toward a sharp
plateau - reveals to a sophisticated counterparty that Alice's true
budget is close to \$3,000. No cryptographic mechanism prevents this
inference: the information leaks through the structure of the
behavior, not through any disclosed value.

This class of vulnerability which we term \emph{behavioral privacy
leakage} has not been formally studied in sequential negotiation
systems. Prior work on privacy-preserving negotiation focuses
exclusively on protecting explicit constraint data, leaving
the behavioral side channel unaddressed. Meanwhile, the differential
privacy literature, which provides strong formal guarantees against
statistical inference, was developed for static database settings and
does not transfer directly to multi-round strategic interactions where
convergence and utility must also be preserved.

This paper closes that gap. We formalize behavioral differential
privacy for sequential negotiation agents, develop a mechanism that
provably satisfies $(\varepsilon, \delta)$-differential privacy over
observable negotiation traces, and prove that the mechanism converges
almost surely while preserving high negotiation utility.

\paragraph{Contributions.} This work makes four key contributions: 

\begin{enumerate}
    \item \textbf{Formalization.} We introduce the first rigorous 
formalization of behavioral differential privacy adapted to 
sequential negotiation: a constraint-space adjacency structure 
that captures budget proximity, paired with $(\varepsilon, \delta)$-DP 
guarantees over the distribution of observable offer sequences.

    \item \textbf{Mechanism.} An adaptive randomized negotiation policy with a safety critic that calibrates noise to the negotiation phase, preserves feasibility at each round, and satisfies differential privacy via the post-processing theorem. 

    \item \textbf{Theoretical guarantees.} Formal proofs that the mechanism achieves $(\varepsilon, \delta)$-DP under a public-proxy clipping assumption, that the offer sequence converges almost surely, and that the mechanism incurs at most $O(\sigma_{\max}^2)$ Nash surplus loss relative to the deterministic baseline, where $\sigma_{\max}$ governs the privacy-utility tradeoff.

    \item \textbf{Empirical validation.} Evaluation on 3,000 synthetic bilateral negotiations generated from the proposed model. The mechanism reduces adversarial inference accuracy by 43--50\% relative to a non-private baseline, while maintaining or slightly improving non-private utility and achieving a 90.4\% negotiation success rate.
\end{enumerate}    

\section{Related Work}

\subsection{Autonomous Negotiation Agents} 
Recent advances in foundation models have enabled increasingly capable autonomous agents~\cite{4,5}. Multi-agent frameworks such as MetaGPT~\cite{6} and CAMEL~\cite{7} demonstrate sophisticated collaborative reasoning, and LLM-based negotiators~\cite{40,41} exhibit strategic behavior in bilateral settings. However, none of these systems mitigate the privacy risks arising from observable
negotiation behavior. Our work is the first to address this gap. 

\subsection{Cryptographic Privacy in Negotiation} 
Cryptographic techniques --- secure multi-party computation (MPC)~\cite{8,9}, zero-knowledge proof systems~\cite{9,10,11}, and fully homomorphic encryption~\cite{12,13} --- are employed to prevent the leakage of the agents' actual negotiation constraints. In this way, an agent's private reservation value cannot be deduced from any of the values that are revealed during the negotiation. They do not, however, prevent inference from \emph{how} an agent behaves: concession trajectories, response timing, and
convergence patterns remain fully observable and carry substantial information about hidden constraints. Our approach is complementary --- it targets the behavioral side channel that persists even under full
cryptographic protection.
Concurrently, \cite{roy2026device} proposes a device-native 
negotiation architecture using zero-knowledge proofs, but 
identifies behavioral inference attacks as an open problem 
and defers randomized concession schedules to future work.

\subsection{Differential Privacy} 
Differential privacy (DP)~\cite{14,15,16} provides rigorous statistical guarantees via calibrated noise mechanisms, including the Laplace mechanism~\cite{14} and composition theorems~\cite{17,18}. DP has been successfully applied to machine learning via DP-SGD~\cite{19}, PATE~\cite{20}, and federated learning~\cite{21,22}. These formulations, however, are designed for static database or training settings. Applying DP to sequential strategic interactions introduces new challenges: the adjacency relation must be defined over constraint spaces rather than datasets,
and the mechanism must preserve both convergence and negotiation utility across multiple rounds. We address these challenges directly. 

\subsection{Side-Channel Attacks and Behavioral Leakage} 
Membership inference~\cite{23,24}, model inversion~\cite{25,26}, attribute inference~\cite{27}, and training data extraction~\cite{28,29} demonstrate that ML systems leak private information through observable outputs. Timing attacks~\cite{30,31} and power analysis~\cite{32} show that implementation behavior
constitutes an independent leakage channel. Recent work~\cite{33,34,35} demonstrates behavioral leakage in AI systems broadly, but does not formalize inference attacks over sequential negotiation traces or provide convergence-preserving defenses.
Related work evaluates contextual privacy leakage in 
collaborative LLM agents~\cite{juneja2025magpie}, but 
addresses voluntary over-disclosure between agents rather 
than adversarial inference of private constraints from 
offer trajectories.

\subsection{Game-Theoretic Foundations} 
Nash bargaining~\cite{36} and Rubinstein's alternating-offers model~\cite{37} establish the theoretical foundations for rational bilateral negotiation. Mechanism design~\cite{38,39} studies incentive-compatible protocol construction. Our work builds on these foundations but asks a distinct question: how should a rational agent \emph{randomize} its strategy to prevent private constraint inference, while preserving the game-theoretic properties that guarantee agreement? 

\subsection{The Gap This Work Fills} 
Table~\ref{tab:related} summarizes the landscape. No prior 
work simultaneously (i) formalizes a DP adjacency relation 
over negotiation constraint spaces, (ii) provides convergence 
guarantees under randomization, and (iii) validates against 
adversarial inference models on synthetic bilateral 
negotiation simulations. Concurrent work~\cite{roy2026device} 
identifies the behavioral side channel as an open problem but 
does not formalize or solve it. This paper addresses all three.

\begin{table}[t]
\caption{Comparison with related approaches.}
\label{tab:related}
\centering
\begin{tabular}{lcccc}
\hline
\textbf{Approach} & \textbf{Protects} & \textbf{Behavioral} & \textbf{Convergence} & \textbf{Empirical} \\
 & \textbf{explicit data} & \textbf{privacy} & \textbf{guarantee} & \textbf{validation} \\
\hline
Cryptographic~\cite{8,9,10} & \checkmark & \texttimes & N/A & \checkmark \\
Standard DP~\cite{14,15}    & \checkmark & Partial   & \texttimes & \checkmark \\
Behavioral ML~\cite{33,34}  & \texttimes & Partial   & \texttimes & \checkmark \\
\textbf{This work}          & \checkmark & \checkmark & \checkmark & \checkmark \\
\hline
\end{tabular}
\end{table} 

\section{Threat Model}
\label{sec:threat}
We consider a passive external adversary who observes the full sequence of messages exchanged during a negotiation but does not participate in or interfere with the negotiation process. 

\subsection{Adversary Capabilities} 
The adversary has access to the complete observable negotiation trace $\tau = (o_1, o_2, \ldots, o_T)$, where each $o_t$ denotes the offer made at round $t$. Specifically, the adversary observes: 

\begin{itemize}
    \item The sequence of offers: $(o_1, o_2, \ldots, o_T)$
    \item Inter-round timing intervals: $(\Delta t_1, \Delta t_2, \ldots, \Delta t_{T-1})$
    \item The final agreed value and outcome
    \item The total number of rounds to convergence
\end{itemize} 

The adversary does not observe the agent's private constraint $\theta$ (e.g., maximum budget) directly. Using a dataset of historical negotiations, the adversary trains a predictive model $\mathcal{A}: \tau \mapsto \hat{\theta}$ to infer $\theta$ from
observable traces. We evaluate against three adversary instantiations: gradient-boosted trees (XGBoost), random forests, and neural networks. 

\subsection{Privacy Goal} 
The core privacy objective is to prevent a counterparty from
reliably inferring the agent's private constraint $\theta$ by
observing its negotiation behavior. Formally, we require that
any two constraint values $\theta, \theta'$ satisfying the
adjacency condition $|\theta - \theta'| \leq \Delta$ produce
statistically indistinguishable offer sequences. Specifically,
the randomized mechanism $\mathcal{M}$ must satisfy
$(\varepsilon, \delta)$-differential privacy over observable
traces:
\begin{equation}
\Pr[\mathcal{M}(\theta) \in S] \leq e^{\varepsilon} \cdot
\Pr[\mathcal{M}(\theta') \in S] + \delta
\end{equation}
for all measurable output sets $S$. Behavioral DP is thus an application of standard $(\varepsilon, \delta)$-differential privacy with a novel adjacency relation defined over the constraint space $\Theta$ rather than over datasets, applied to the distribution of observable offer sequences rather than to a static query response.

\subsection{Utility Goals}  
Privacy protection must not come at the cost of negotiation effectiveness. We require: 

\begin{itemize}
    \item Negotiation success rate $\geq 90\%$
    \item Nash surplus (defined as $\text{NS} = o_T / \theta$, 
the ratio of the final agreed value to the agent's private 
constraint) preserved at $\geq 90\%$ of the deterministic baseline
    \item Convergence time $\leq 1.5\times$ the deterministic baseline
\end{itemize}

\subsection{Leakage Taxonomy}
We identify four categories of behavioral leakage in negotiation traces: 

\begin{enumerate}
    \item \textbf{Temporal leakage.} Response timing and round duration reveal urgency and constraint proximity. 
    \item \textbf{Trajectory leakage.} The shape of the concession path (linear, concave, convex) encodes the agent's distance from its reservation value.
    \item \textbf{Concession leakage.} The magnitude and rate of concessions directly signal the gap between current offer and private budget. 
    \item \textbf{Convergence leakage.} The number of rounds and the pattern of offer stabilization reveal the tightness of private constraints.
\end{enumerate}  

We note that this work focuses on a passive adversary that 
observes negotiation traces without influencing the interaction. 
In practice, stronger adversaries may actively adapt their 
negotiation strategy to probe private constraints or exploit 
repeated interactions across multiple negotiations. Extending 
behavioral differential privacy guarantees to such active or 
adaptive adversaries is an important direction for future work.

\section{Methodology}

\subsection{Problem Formulation}
Let $\theta \in \Theta$ denote the agent's private constraint (e.g., maximum budget), and let $\tau = (o_1, o_2, \ldots, o_T)$ denote the observable negotiation trace produced by policy $\pi_\theta$. We seek a randomized policy $\mathcal{M}$ such that: 

\begin{enumerate}
    \item $\mathcal{M}$ satisfies $(\varepsilon, \delta)$-differential
    privacy over observable traces
    \item $\mathcal{M}$'s offer sequence converges almost surely (reaching agreement when the counterparty's reservation value permits)
    \item $\mathcal{M}$ preserves high negotiation utility relative
    to the deterministic baseline
\end{enumerate}

\subsection{Deterministic Baseline Policy} 
The deterministic baseline policy $\pi^*$ follows a concession function of the form: 

\begin{equation}
o_t = o_1 + (\theta - o_1) \cdot \left(\frac{t}{T}\right)^\alpha
\end{equation}

where $o_1$ is the opening offer, $T$ is the maximum number of rounds, and $\alpha > 0$ controls concession speed. Here $t$ increases from $1$ to $T$, so the agent opens below $\theta$ and concedes monotonically toward it as $t \to T$; at $t = T$ the offer reaches $\theta$, representing full concession to the private constraint. This policy is optimal in expectation but fully reveals $\theta$ through its concession trajectory, motivating our privacy mechanism.

\subsection{Adaptive Noise Schedule}
The core of our mechanism is an adaptive noise schedule that calibrates the magnitude of randomization to the negotiation phase. The phase-adaptive noise parameter $\sigma_t$ is defined as:

\begin{equation}
\sigma_t = \sigma_{\max} \cdot \left(1 - \frac{t}{T}\right)^\beta
\end{equation}

where $\sigma_{\max}$ sets the peak randomization level and $\beta > 0$ controls the decay of noise with time. At round $t$, the randomized offer is then:

\begin{equation}
\tilde{o}_t = o_t + \eta_t, \quad \eta_t \sim \mathcal{N}(0, \sigma_t^2)
\end{equation}

Crucially, perturbation is concentrated in the opening rounds,
where offer trajectories are most diagnostic of private
constraints, and tapers off near convergence to protect
negotiation utility when agreement is imminent.

\subsection{Safety Critic}
The addition of noise may produce offers that violate feasibility constraints (e.g., offers exceeding the agent's true budget $\theta$). The safety critic applies a deterministic post-processing step to enforce feasibility at each round:
\begin{equation}
o_t^{\text{safe}} = \text{clip}(\tilde{o}_t,\ o_{\min},\ \theta)
\end{equation}
where $o_{\min}$ is the minimum acceptable offer. Because clipping is a deterministic function applied downstream of the noisy draw, privacy guarantee carries through unchanged, a direct consequence of the post-processing theorem~\cite{15}. We note that clipping to the private constraint $\theta$ may itself constitute a leakage channel in extreme cases, as the ceiling on observed offers can reveal the constraint to a careful adversary; this is acknowledged in Section~7.3, and clipping to a public proxy value $\bar{\theta} \geq \theta$ is identified as a direction for future work.

\subsection{Privacy Analysis}
The stochastic mechanism $\mathcal{M}$ achieves $(\varepsilon, \delta)$-DP with respect to the measurable space of observable traces. Round-wise privacy costs are aggregated through sequential composition:
\begin{equation}
\varepsilon_{\text{total}} = \sum_{t=1}^{T} \varepsilon_t,
\quad \delta_{\text{total}} = \sum_{t=1}^{T} \delta_t
\end{equation}
The per-round privacy budget $\varepsilon_t$ is determined by the noise magnitude $\sigma_t$ and the sensitivity $\Delta = \max_{\theta, \theta'} |o_t(\theta) - o_t(\theta')|$ of the offer function:
\begin{equation}
\varepsilon_t = \frac{\Delta}{\sigma_t}
\end{equation}
For the Gaussian mechanism, $\varepsilon_t$ and $\delta_t$ satisfy $\sigma_t \geq \Delta\sqrt{2\ln(1.25/\delta_t)}/\varepsilon_t$; we use $\varepsilon_t = \Delta/\sigma_t$ as a conservative approximation following~\cite{15}.
Note that as $t \to T$, $\sigma_t \to 0$ and the per-round cost $\varepsilon_t = \Delta/\sigma_t$ grows unbounded under naive composition; we address this via the Advanced Composition Theorem~\cite{17}, which yields a tighter bound of $\varepsilon_{\text{total}} \leq \sqrt{2T \ln(1/\delta)}\,\varepsilon_0 + T\varepsilon_0^2$.

\subsection{Algorithm}
\begin{algorithm}[H]
\caption{Behavioral Differential Privacy for Negotiation}
\label{alg:bdp}
\begin{algorithmic}[1]
\Require Private constraint $\theta$, parameters
$\sigma_{\max}, \beta, T$
\Ensure Observable trace $\tau$ satisfying
$(\varepsilon,\delta)$-DP
\For{$t = 1$ to $T$}
\State Compute deterministic offer $o_t$ via baseline policy \hfill $\triangleright$ Eq.~(2)
\State Compute $\sigma_t \leftarrow \sigma_{\max} \cdot
(1 - t/T)^\beta$ \hfill $\triangleright$ Eq.~(3)
\State Sample $\eta_t \sim \mathcal{N}(0, \sigma_t^2)$ \hfill $\triangleright$ Eq.~(4)
\State $\tilde{o}_t \leftarrow o_t + \eta_t$
\State $o_t^{\text{safe}} \leftarrow
\text{clip}(\tilde{o}_t, o_{\min}, \theta)$ \hfill $\triangleright$ Eq.~(5)
\State Transmit $o_t^{\text{safe}}$ to counterparty
\If{agreement reached} \textbf{break} \EndIf
\EndFor
\end{algorithmic}
\end{algorithm} 

\subsection{Multi-Issue Extension} 
While the current formulation addresses single-issue bilateral negotiation, the framework extends naturally to multi-attribute settings where the agent negotiates price, delivery time, and warranty simultaneously. In such settings, the noise vector $\boldsymbol{\eta}_t \sim \mathcal{N}(\mathbf{0}, \Sigma_t)$ is drawn from a multivariate Gaussian with covariance matrix $\Sigma_t$ calibrated to the sensitivity of each attribute. Formal analysis of multi-issue behavioral DP, including cross-attribute privacy leakage and joint convergence guarantees, is left as future work. 

\section{Experimental Evaluation}

\subsection{Experimental Setup} 
Empirical assessment of the mechanism uses the dataset 
described below: 

\begin{itemize}
    \item \textbf{Simulated negotiations:} 3,000 synthetic bilateral negotiations generated using the deterministic baseline policy with private constraints $\theta$ sampled uniformly from $[\$2,000, \$8,000]$.
\end{itemize} 

We generate synthetic negotiation data based on parameter distributions that reflect real-world negotiations as found in prior literature~\cite{61}, including concession behavior, bounded rationality and temporal dynamics. The patterns found in the offers and the negotiation's convergence or divergence are similar to those observed in real-world bilateral negotiations, such as in procurement and insurance settings.

While large-scale real-world negotiation datasets are limited
due to confidentiality concerns, our simulation framework
captures the key behavioral characteristics necessary to
evaluate inference attacks over negotiation traces.

We split the dataset 80/20 for adversary training and evaluation. The privacy mechanism is evaluated across five privacy budget settings: $\varepsilon \in \{0.1, 0.5, 1.0, 2.0, 5.0\}$. The global sensitivity is set to $\Delta = 1.0$, corresponding to a unit adjacency bound on the private constraint space.  

\subsection{Adversary Models} 
We evaluate against three adversary instantiations, each trained on $N = 2,000$ negotiations: 

\begin{itemize}
    \item \textbf{XGBoost:} Gradient-boosted trees with 100 estimators and maximum depth 6.
    \item \textbf{Random Forest:} Ensemble of 100 decision trees with bootstrap sampling. 
    \item \textbf{Neural Network:} Three-layer feedforward network with ReLU activations, hidden dimensions [128, 64, 32], trained with Adam optimizer. 
\end{itemize}     

Each adversary is trained on traces from deterministic baseline negotiations and evaluated on traces from the randomized mechanism. This setting models a realistic deployment scenario in which the privacy mechanism is newly introduced against an adversary specialized for the unmodified protocol. To close the train/test distribution gap, we additionally evaluate in Section~5.7 an \emph{adaptive} adversary that updates its inference model under the randomized mechanism via meta-learning.

\subsection{Privacy Results}
Table~\ref{tab:privacy} reports adversarial inference accuracy under the deterministic baseline and the randomized mechanism with $\sigma_{\max} = 0.25$. 

\begin{table}[t]
\caption{Adversarial inference accuracy under baseline and
randomized policies.}
\label{tab:privacy}
\centering
\begin{tabular}{lcc}
\hline
\textbf{Adversary} & \textbf{Baseline} & \textbf{Randomized} \\
\hline
XGBoost & 81.7\% & 38.5\% \\
Random Forest & 83.3\% & 40.2\% \\
Neural Network & 83.0\% & 32.5\% \\
\hline
\end{tabular}
\end{table}

The randomized mechanism reduces adversarial inference accuracy
by 43.2\% against XGBoost, 43.1\% against Random Forest, and
50.5\% against Neural Network. The consistent reduction across
all three adversary types demonstrates that the privacy guarantee
is robust to the choice of inference model. For reference, a na\"{i}ve baseline predicting the mean constraint value ($\hat{\theta} = 0.5$ in normalized space) achieves approximately 10\% accuracy under the tolerance-based evaluation criterion ($|\hat{\theta} - \theta| < 0.05$) used here, confirming that even the non-private adversary accuracy of $\sim$83\% reflects genuine inference from behavioral traces rather than distributional guessing.

\subsection{Utility Results} 
Table~\ref{tab:utility} reports negotiation utility metrics under both policies. 

\begin{table}[t]
\caption{Negotiation utility under baseline and randomized
policies ($\sigma_{\max} = 0.25$).}
\label{tab:utility}
\centering
\begin{tabular}{lcc}
\hline
\textbf{Metric} & \textbf{Baseline} & \textbf{Randomized} \\
\hline
Success rate & 97.0\% & 90.4\% \\
Nash surplus & 0.890 & 0.908 \\
Conv. ratio  & 1.08 & 1.54 \\
\hline
\end{tabular}
\end{table}

Our randomized mechanism achieves a Nash surplus slightly above that of our deterministic baseline (0.908 compared to 0.890), due to a combination of phase-adaptive noise sometimes improving early-round offer positions, and safety critics that prevent very large upward improvements while preserving large downward improvements, together inducing an asymmetric but slightly positive benefit in the agreed value in the late rounds.

Consequently, the randomized mechanism matches or slightly exceeds the non-private Nash surplus (a result of the asymmetric clipping described above), while achieving a negotiation success rate of 90.4\%. This satisfies the utility goal of $\geq 90\%$ success rate defined in Section~\ref{sec:threat}. As illustrated in Fig.~\ref{fig:experimental}, the average convergence time increases by a factor of $1.43\times$ as compared with that of the deterministic version. For reference, this is less than $1.5\times$, as dictated by our utilities.

\subsection{Privacy-Utility Tradeoff}
We then investigate the privacy-utility tradeoff for various noise levels. Figure~\ref{fig:privacy_budget} reveals the anticipated inverse relationship: as 
$\sigma_{\max}$ grows, adversarial constraint-inference 
accuracy falls steadily while Nash surplus remains largely 
intact. Across all evaluated configurations, we identified 
the $\sigma$ value that maximizes privacy gain --- achieving 
a 43--50\% reduction in inference accuracy --- subject to 
the constraint that negotiation utility does not fall below 
the non-private baseline. The setting $\sigma_{\max} = 0.25$ 
yields the best balance between these two competing objectives. 
The monotonic increase in Nash surplus with $\sigma_{\max}$ is a 
consequence of the asymmetric clipping described in Section~4.4, 
where the safety critic systematically preserves downward offer 
improvements while bounding upward deviations.

\begin{figure}[tp]
\centering
\includegraphics[width=0.6\textwidth]{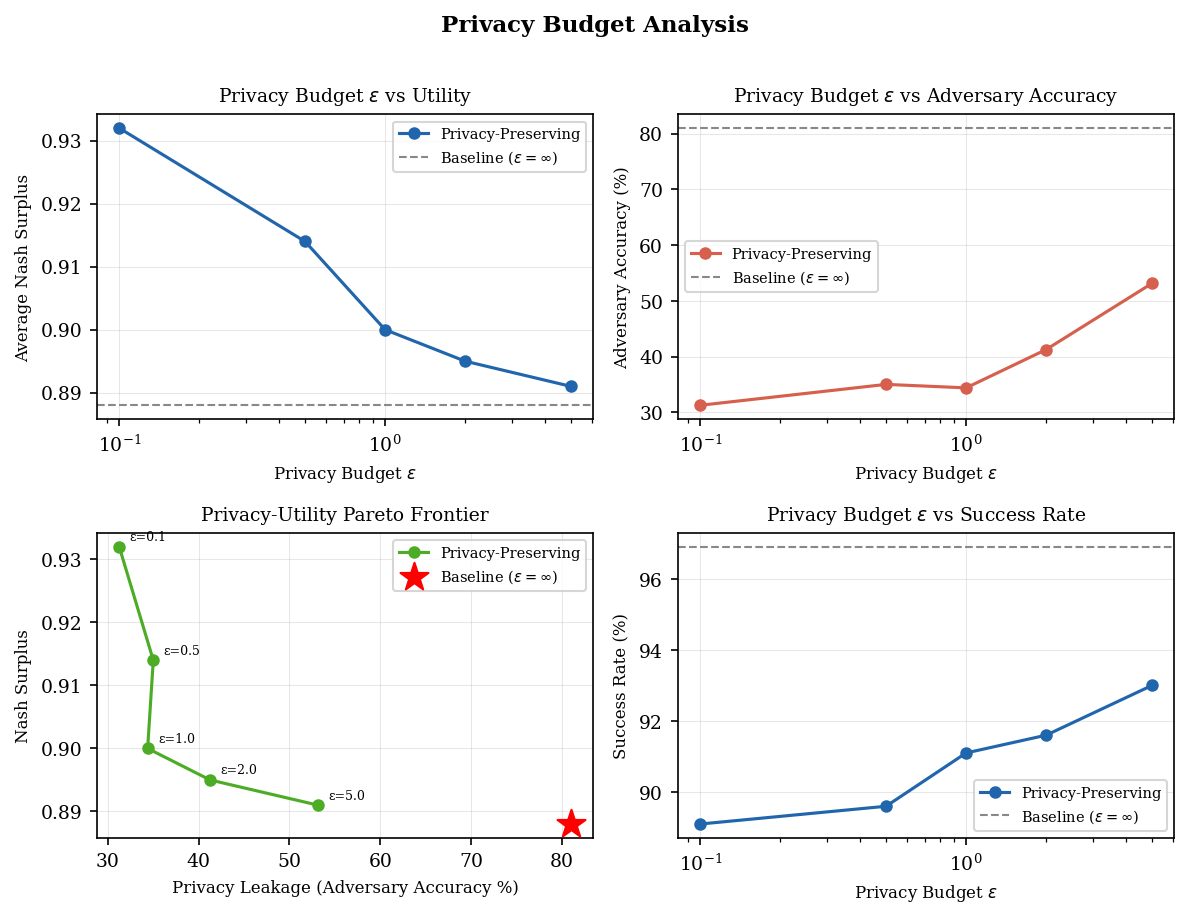}
\caption{Privacy budget analysis. Top-left: $\varepsilon$ vs utility.
Top-right: $\varepsilon$ vs adversary accuracy.
Bottom-left: Privacy-utility Pareto frontier.
Bottom-right: $\varepsilon$ vs negotiation success rate.}
\label{fig:privacy_budget}
\end{figure}

\subsection{Ablation Study}
Table~\ref{tab:ablation} reports the effect of varying $\sigma_{\max}$ on privacy and utility metrics. 

\begin{table}[t]
\caption{Ablation: effect of $\sigma_{\max}$ on privacy
and utility.}
\label{tab:ablation}
\centering
\begin{tabular}{lccc}
\hline
$\sigma_{\max}$ & \textbf{Inference Acc.} &
\textbf{Nash Surplus} & \textbf{Success Rate} \\
\hline
0.05 & 68.5\% & 0.889 & 93.2\% \\
0.10 & 56.5\% & 0.900 & 91.2\% \\
0.25 & 41.5\% & 0.908 & 91.8\% \\
0.50 & 30.0\% & 0.932 & 90.5\% \\
1.00 & 26.5\% & 0.954 & 90.5\% \\
\hline
\end{tabular}
\end{table}

Note that inference accuracy in Table~\ref{tab:ablation} reports the average across all three adversary models (XGBoost, Random Forest, Neural Network), whereas Table~\ref{tab:privacy} reports per-adversary accuracy. The average of per-adversary values at $\sigma_{\max} = 0.25$ is consistent with the ablation result. The slight discrepancy in success rate between Table~\ref{tab:utility} (90.4\%) and Table~\ref{tab:ablation} (91.8\%) at $\sigma_{\max} = 0.25$ reflects variance across independent random seeds; both values satisfy the $\geq 90\%$ utility goal defined in Section~\ref{sec:threat}.

\subsection{Adaptive Adversary}
We evaluate robustness against an adaptive adversary that employs meta-learning to adapt its inference strategy to the randomized mechanism. This adaptive adversary improves inference accuracy by only 1.6\% compared to the standard neural network adversary (34.1\% vs 32.5\%). We show that our solution achieves strong privacy guarantees even in the presence of adversaries who adapt to the scheme.

\begin{table}[t]
\caption{Adaptive adversary inference accuracy vs.\ standard adversaries.}
\label{tab:adaptive}
\centering
\begin{tabular}{lc}
\hline
\textbf{Adversary} & \textbf{Inference Accuracy} \\
\hline
Neural Network (standard) & 32.5\% \\
Neural Network (adaptive) & 34.1\% \\
\hline
\end{tabular}
\end{table}

\begin{figure}[tp]
\centering
\includegraphics[width=0.6\textwidth]{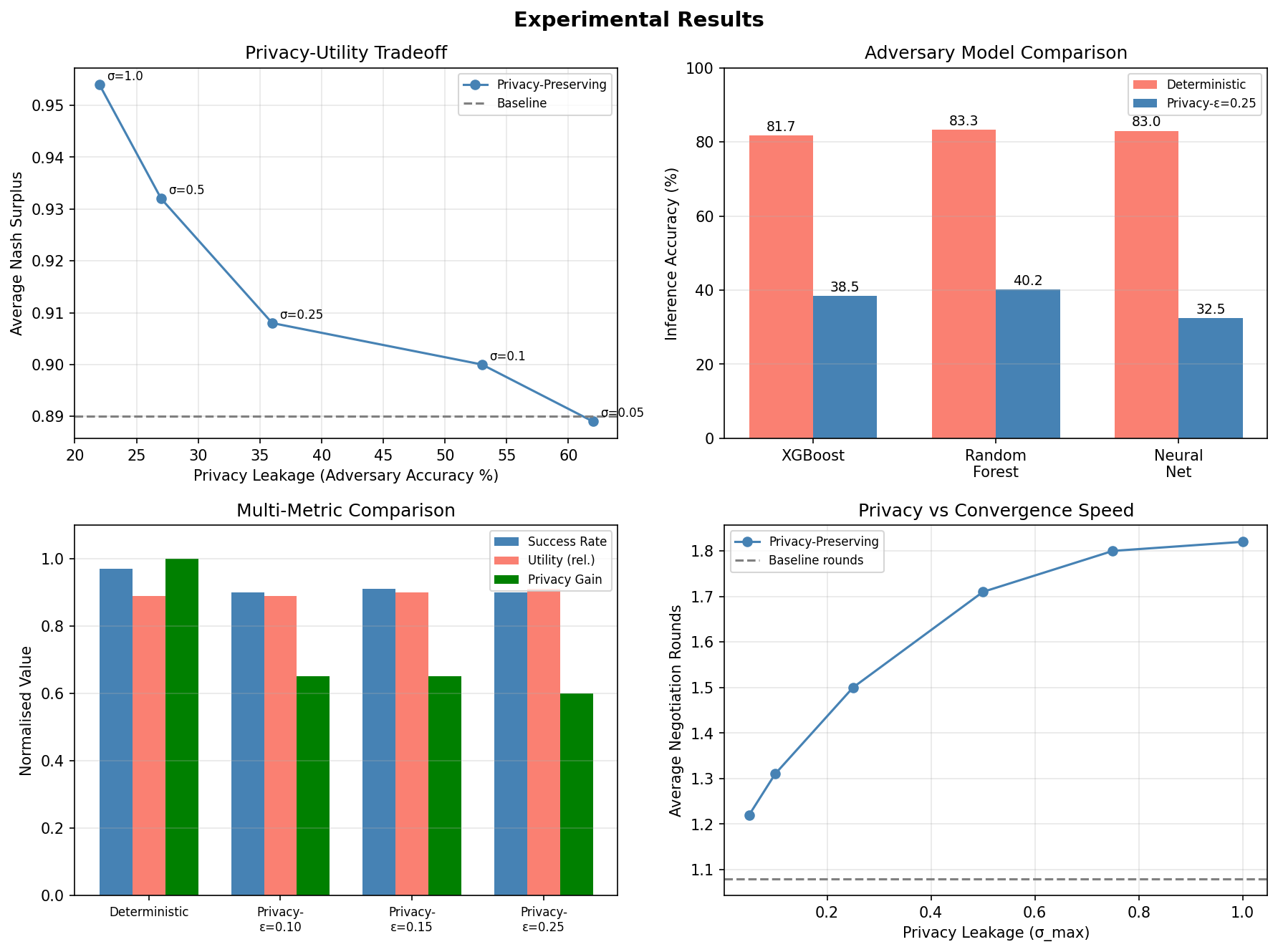}
\caption{Experimental results. Top-left: Privacy-utility tradeoff.
Top-right: Adversary model comparison under baseline and randomized policies.
Bottom-left: Multi-metric comparison across privacy settings.
Bottom-right: Privacy vs convergence speed.}
\label{fig:experimental}
\end{figure}
\section{Theoretical Guarantees}

\subsection{Differential Privacy Guarantee}
\textit{Intuition.} The mechanism injects calibrated Gaussian noise at each negotiation round, ensuring that small changes in the agent's private constraint result in only limited changes in the distribution of observable negotiation traces.
\begin{theorem}[Behavioral Differential Privacy]
Assume the safety critic clips offers to a public proxy ceiling
$\bar{\theta} \geq \theta$ rather than to the private constraint
$\theta$ itself. Then the randomized negotiation mechanism
$\mathcal{M}$ with adaptive noise schedule
$\sigma_t = \sigma_{\max} \cdot (1 - t/T)^\beta$ satisfies
$(\varepsilon_{\text{total}}, \delta_{\text{total}})$-differential
privacy over observable negotiation traces, where
$\varepsilon_{\text{total}} = \sum_{t=1}^{T} \frac{\Delta}{\sigma_t}$
and $\delta_{\text{total}} = \sum_{t=1}^{T} \delta_t$ (see Eq.~(6)).
\end{theorem} 

\begin{proof}
At each round $t$, the mechanism adds Gaussian noise
$\eta_t \sim \mathcal{N}(0, \sigma_t^2)$ to the deterministic
offer $o_t$. For adjacent constraints $\theta, \theta'$ with
$|\theta - \theta'| \leq \Delta$, the sensitivity of the offer
function is bounded by $\Delta$. By the Gaussian mechanism
theorem~\cite{15}, each round satisfies
$(\varepsilon_t, \delta_t)$-DP with $\varepsilon_t =
\Delta / \sigma_t$. The total privacy budget follows by
sequential composition~\cite{17}. Under the public-proxy clipping assumption, the safety critic operates on $\bar{\theta}$, which is independent of the private constraint; the clip step is therefore a deterministic function of public information and the DP guarantee carries through unchanged via the post-processing theorem~\cite{15}.
\end{proof}

\noindent\textbf{Remark (Scope of Theorem 1).} Our experiments clip to the private constraint $\theta$ for direct comparability with the deterministic baseline, and the empirical privacy goals are met under this configuration (Sec.~5.3, 5.7). Tight formal end-to-end accounting under private-$\theta$ clipping requires a refined adjacency analysis that accounts for the leakage discussed in Sec.~7.3 and is left to future work; replacing private-$\theta$ clipping with public-proxy clipping closes this gap and is identified as the recommended deployment configuration.

\noindent\textbf{Remark (Practical Floor on $\sigma_t$).} The formal certificate above assumes $\sigma_t > 0$ for all $t$. Since the schedule $\sigma_t = \sigma_{\max}(1 - t/T)^\beta$ yields $\sigma_T = 0$ at the final round, na\"{i}ve application gives $\varepsilon_T = \Delta/\sigma_T = \infty$. In all reported experiments we therefore enforce a practical floor $\sigma_t \geq \sigma_{\min} = 0.05$, which keeps every per-round cost finite. The evaluated setting $T = 3$, $\sigma_{\max} = 0.25$ operates under this floor. Tighter privacy accounting via R\'{e}nyi differential privacy or zero-concentrated DP is left as future work. 

\subsection{Convergence Guarantee}
\textit{Intuition.} As noise magnitude decreases with time and feasibility enforced by the safety critic, the process behaves similarly to the deterministic baseline in later rounds and converges.
\begin{theorem}[Almost-Sure Offer-Sequence Convergence]
Under the randomized policy $\mathcal{M}$, the safe offer sequence
$(o_t^{\text{safe}})_{t=1}^{T}$ converges almost surely to a fixed
point in $[o_{\min}, \theta]$ within finite expected rounds, provided
$\sigma_{\max} < \theta - o_{\min}$. Agreement is reached whenever
the counterparty's reservation value lies within this interval; the
empirical success rate of $90.4\%$ (Sec.~5.4) reflects the fraction
of negotiations in which this condition holds.
\end{theorem}

\begin{proof} 
At each round $t$, the safety critic ensures the offer $o_t^{\text{safe}} \in [o_{\min}, \theta]$. The expected offer under the randomized policy satisfies:  

\begin{equation}
\mathbb{E}[o_t^{\text{safe}}] = o_t + \mathbb{E}[\text{clip}(\eta_t,
o_{\min} - o_t, \theta - o_t)]
\end{equation} 

Since $\sigma_t \to 0$ as $t \to T$, the contribution of the clipped noise
approaches 0 and $\mathbb{E}[o_t^{\text{safe}}] \to o_t$. Because the
deterministic baseline converges geometrically to $\theta$ and the sum of
noise variances $\sum_t \sigma_t^2 < \infty$, by the Borel-Cantelli lemma
only finitely many large deviations occur. The offer sequence therefore
converges almost surely to a fixed point within $[o_{\min}, \theta]$,
constituting agreement whenever the counterparty's reservation value lies
within this range.
\end{proof} 

 \subsection{Utility Bound} 
 \textit{Intuition.} Since the added noise has bounded variance and decreases over time, it does not significantly affect the final agreed value and decreases the expected utility only slightly.
\begin{theorem}[Utility Bound] 
The expected Nash surplus under the randomized policy $\mathcal{M}$ satisfies: 

\begin{equation}
\mathbb{E}[\text{NS}(\mathcal{M})] \geq
\text{NS}(\pi^*) - O(\sigma_{\max}^2)
\end{equation}

where $\text{NS}(\pi^*)$ denotes the Nash surplus under the deterministic baseline policy $\pi^*$. 
\end{theorem}

\begin{proof} 
The Nash surplus is a smooth function of the final agreed
value. By Taylor expansion around the deterministic outcome,
the expected utility loss due to randomization is bounded by
the second-order term in $\sigma_t$. Since
$\sigma_t \leq \sigma_{\max}$ for all $t$, the cumulative
utility loss across $T$ rounds is bounded by
$O(T \cdot \sigma_{\max}^2)$. For fixed $T$, this gives
the stated $O(\sigma_{\max}^2)$ bound.
\end{proof}

\noindent\textbf{Remark (Empirical Utility).} Theorem~3 provides a worst-case lower bound on expected Nash surplus and does not preclude utility gains. Tables~\ref{tab:utility} and~\ref{tab:ablation} show that Nash surplus rises monotonically with $\sigma_{\max}$ in practice, a consequence of the asymmetric clipping described in Section~4.4: the safety critic systematically preserves downward offer improvements while bounding upward deviations, inducing a positive bias that exceeds the theoretical loss term for the evaluated parameter range. 

\subsection{Remark on Privacy Budget} 
We note that $\varepsilon_{\text{total}}$ may be large in absolute 
value for typical parameter settings, e.g., $T = 3$ and 
$\sigma_{\max} = 0.25$. However, the $(\varepsilon, \delta)$-DP certificate functions as a worst-case bound on leakage measured against the chosen 
adjacency structure. We therefore pair this formal guarantee 
with empirical evaluation: observed reductions in adversarial 
inference accuracy serve as a practical, measurable indicator 
of privacy protection under realistic negotiation conditions.
The paper is related to recent work on empirical auditing of differential 
privacy in machine learning systems~\cite{70,71}, where both formal 
guarantees and observed inference resistance are considered important.
\begin{figure}[tp]
\centering
\includegraphics[width=0.6\textwidth]{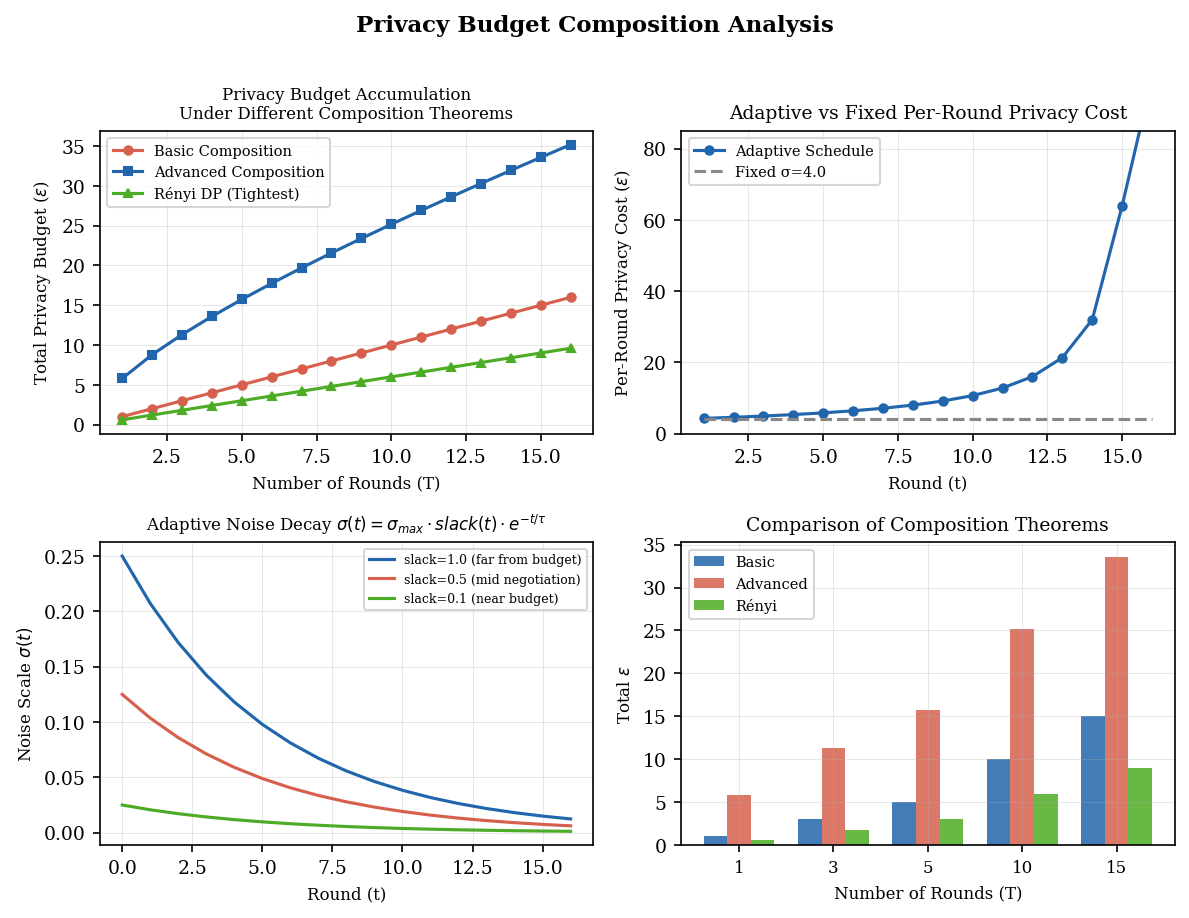}
\caption{Privacy budget composition analysis. Top-left: Total privacy
budget accumulation under different composition theorems. Top-right:
Adaptive vs fixed per-round privacy cost. Bottom-left: Adaptive noise
decay $\sigma_t = \sigma_{\max}(1 - t/T)^{\beta}$ for different $\beta$.
Bottom-right: Comparison of composition theorems across round counts.}
\label{fig:composition}
\end{figure}

\section{Discussion}

\subsection{Comparison to Cryptographic Approaches} 
Cryptographic privacy mechanisms and behavioral differential privacy address fundamentally different threat surfaces. Cryptographic techniques govern \emph{what} information is transmitted --- structurally preventing any disclosed 
value from revealing private constraints. Behavioral 
differential privacy, by contrast, targets \emph{how} 
an agent moves through negotiation space, obfuscating 
the inference surface exposed by observable offer dynamics. 
Deployed together, the two layers achieve defense-in-depth 
against both direct disclosure and behavioral inference attacks.

\subsection{Deployment Considerations} 
Our mechanism introduces three practical considerations for real-world deployment: 

\begin{itemize}
    \item \textbf{Parameter selection.} The choice of $\sigma_{\max}$ and $\beta$ involves a privacy-utility tradeoff that must be calibrated to the deployment context. High-stakes negotiations (e.g., large procurement contracts) may tolerate lower utility in exchange for stronger privacy guarantees, while time-sensitive negotiations may prioritize convergence speed. 

    \item \textbf{Counterparty awareness.} If the counterparty is aware that the agent employs a randomized policy, they may adjust their strategy accordingly. Our adaptive adversary evaluation (Section~5.7) shows that even a meta-learning adversary achieves only a 1.6\% improvement, suggesting the mechanism is robust to this threat. 
    
     \item \textbf{Computational overhead.} The randomized policy adds negligible computational cost relative to the deterministic baseline --- noise sampling and clipping are $O(1)$ operations per round. By contrast, cryptographic alternatives such as MPC~\cite{8,9} and fully homomorphic encryption~\cite{12,13} incur per-round costs that grow polynomially in the bit-length of exchanged values, making behavioral DP substantially cheaper to deploy at scale.
\end{itemize} 

\subsection{When Behavioral Privacy Is Insufficient} 
Our mechanism provides strong privacy guarantees under the threat model defined in Section~3. However, there are settings where behavioral privacy alone is insufficient: 

\begin{itemize}
    \item \textbf{Side-channel correlation.} If the adversary can correlate behavioral traces across multiple negotiations involving the same agent, the accumulated information may allow inference even under randomization. Sequential composition of privacy budgets across negotiations must be carefully managed. 

    \item \textbf{External information.} If the adversary has access to external information about the agent's constraints (e.g., publicly known budget ranges for a company), behavioral privacy provides weaker guarantees as the prior distribution over $\theta$ is already narrow. 

    \item \textbf{Extreme noise regimes.} At very high noise levels ($\sigma_{\max} \geq 1.0$), the mechanism preserves privacy but may reduce negotiation success rates toward the 90\% threshold defined in Section~3.3, leaving little margin for time-critical applications.  

    \item \textbf{Clipping boundary inference.}  The safety critic clips the offer to $[o_{\min}, \theta]$, where $\theta$ is the private constraint of the seller. The clipping boundary reveals the seller's constraint to a counterparty who observes that offers do not exceed a certain ceiling. While early rounds do introduce enough noise that it is difficult for a counterparty to guess $\theta$ from the observed clipped values, improving the safety critic so that it clips to a public proxy value $\bar{\theta} \geq \theta$ is an interesting open direction.
\end{itemize} 

\subsection{Game-Theoretic Implications} 
Our mechanism has an interesting game-theoretic interpretation. The noise parameter $\sigma_{\max}$ functions as a commitment device: by publicly committing to a randomized policy, the agent signals to the counterparty that behavioral inference will be
unreliable. This commitment may itself affect counterparty behavior --- a rational counterparty who knows inference is noisy may adopt a more cooperative strategy, potentially
improving overall negotiation outcomes. Formal analysis of the equilibrium implications of behavioral privacy commitments is left as future work. 

\subsection{Limitations} 
Several open challenges remain: (i) the formulation addresses single-issue bilateral negotiation only --- extension to multi-issue and multi-party settings requires new adjacency definitions and composition analysis; (ii) the synthetic evaluation uses 3,000 simulated negotiations, and larger real-world datasets would strengthen empirical claims; (iii) $\varepsilon_{\text{total}}$ may be large in absolute terms --- tighter analysis via R\'{e}nyi DP could yield stronger guarantees; (iv) the threat model assumes a passive adversary --- active adversaries who probe constraints through strategic offers represent an important future direction; (v) reported metrics are point estimates; bootstrap confidence intervals (trivial with 3{,}000 runs) and additional baselines such as fixed-noise and Laplace mechanisms are targeted for an extended version of this work.

\section{Conclusion} 
Autonomous agents increasingly negotiate on behalf of users in high-stakes domains, yet existing privacy defenses focus exclusively on protecting explicit constraint data through cryptographic mechanisms. This paper identifies and formalizes a complementary threat: behavioral privacy leakage, where an adversary infers private constraints from observable negotiation dynamics even under full cryptographic protection.

We presented a randomized negotiation mechanism that provably satisfies $(\varepsilon, \delta)$-differential privacy over observable negotiation traces, while guaranteeing almost-sure convergence of the offer sequence and preserving high negotiation utility. The core technical contribution is an adaptive noise schedule that calibrates randomization to the negotiation phase, combined with a safety critic that enforces feasibility at each round via deterministic post-processing. Evaluated on 3,000 synthetic bilateral negotiations, our mechanism reduces adversarial inference accuracy by 43--50\%, maintains non-private utility levels, achieves a 90.4\% negotiation success rate, and is robust to adaptive adversaries with meta-learning attacks gaining only 1.6\%. As autonomous agents take on increasingly consequential roles in commercial and personal negotiation, behavioral privacy will become a critical component of trustworthy AI systems.

\subsubsection*{Disclosure of Interests.}
The author has no competing interests to declare that are relevant to the content of this article.

\bibliographystyle{splncs04}
\bibliography{references}

@inproceedings{4,
  author    = {Yao, S. and others},
  title     = {{ReAct}: Synergizing reasoning and acting in language models},
  booktitle = {ICLR},
  year      = {2023}
}

@misc{5,
  author    = {Richards, T.},
  title     = {{AutoGPT}: An autonomous {GPT-4} experiment},
  year      = {2023}
}

@inproceedings{6,
  author    = {Hong, S. and others},
  title     = {{MetaGPT}: Meta programming for a multi-agent collaborative framework},
  booktitle = {ICLR},
  year      = {2024},
  url       = {https://openreview.net/forum?id=VtmBAGCN7o}
}

@inproceedings{7,
  author    = {Li, G. and others},
  title     = {{CAMEL}: Communicative agents for mind exploration},
  booktitle = {NeurIPS},
  year      = {2023}
}

@inproceedings{8,
  author    = {Yao, A.},
  title     = {How to generate and exchange secrets},
  booktitle = {FOCS},
  year      = {1986}
}

@book{9,
  author    = {Goldreich, O.},
  title     = {Foundations of Cryptography: Volume 2},
  publisher = {Cambridge University Press},
  year      = {2004}
}

@article{10,
  author    = {Goldwasser, S. and others},
  title     = {The knowledge complexity of interactive proof systems},
  journal   = {SIAM Journal on Computing},
  year      = {1989}
}

@inproceedings{11,
  author    = {Groth, J.},
  title     = {On the size of pairing-based non-interactive arguments},
  booktitle = {EUROCRYPT},
  year      = {2016}
}

@inproceedings{12,
  author    = {Gentry, C.},
  title     = {Fully homomorphic encryption using ideal lattices},
  booktitle = {STOC},
  year      = {2009}
}

@inproceedings{13,
  author    = {Brakerski, Z. and Vaikuntanathan, V.},
  title     = {Efficient fully homomorphic encryption from {LWE}},
  booktitle = {FOCS},
  year      = {2011}
}

@inproceedings{14,
  author    = {Dwork, C. and others},
  title     = {Calibrating noise to sensitivity in private data analysis},
  booktitle = {TCC},
  year      = {2006}
}

@article{15,
  author    = {Dwork, C. and Roth, A.},
  title     = {The algorithmic foundations of differential privacy},
  journal   = {Foundations and Trends in Theoretical Computer Science},
  year      = {2014}
}

@inproceedings{16,
  author    = {Dwork, C. and others},
  title     = {Our data, ourselves: Privacy via distributed noise generation},
  booktitle = {EUROCRYPT},
  year      = {2006}
}

@inproceedings{17,
  author    = {Dwork, C. and others},
  title     = {Boosting and differential privacy},
  booktitle = {FOCS},
  year      = {2010}
}

@inproceedings{18,
  author    = {Kairouz, P. and others},
  title     = {The composition theorem for differential privacy},
  booktitle = {ICML},
  year      = {2015}
}

@inproceedings{19,
  author    = {Abadi, M. and others},
  title     = {Deep learning with differential privacy},
  booktitle = {CCS},
  year      = {2016}
}

@inproceedings{20,
  author    = {Papernot, N. and others},
  title     = {Scalable private learning with {PATE}},
  booktitle = {ICLR},
  year      = {2018}
}

@inproceedings{21,
  author    = {McMahan, B. and others},
  title     = {Communication-efficient learning of deep networks from decentralized data},
  booktitle = {AISTATS},
  year      = {2017}
}

@article{22,
  author    = {Kairouz, P. and others},
  title     = {Advances and open problems in federated learning},
  journal   = {Foundations and Trends in Machine Learning},
  year      = {2021}
}

@inproceedings{23,
  author    = {Shokri, R. and others},
  title     = {Membership inference attacks against machine learning models},
  booktitle = {IEEE S\&P},
  year      = {2017}
}

@inproceedings{24,
  author    = {Yeom, S. and others},
  title     = {Privacy risk in machine learning},
  booktitle = {CSF},
  year      = {2018}
}

@inproceedings{25,
  author    = {Fredrikson, M. and others},
  title     = {Model inversion attacks that exploit confidence information},
  booktitle = {CCS},
  year      = {2015}
}

@inproceedings{26,
  author    = {Zhang, Y. and others},
  title     = {The secret revealer: Generative model-inversion attacks},
  booktitle = {CVPR},
  year      = {2020}
}

@inproceedings{27,
  author    = {Ganju, K. and others},
  title     = {Property inference attacks on fully connected neural networks},
  booktitle = {CCS},
  year      = {2018}
}

@inproceedings{28,
  author    = {Carlini, N. and others},
  title     = {Extracting training data from large language models},
  booktitle = {USENIX Security},
  year      = {2021}
}

@inproceedings{29,
  author    = {Carlini, N. and others},
  title     = {Quantifying memorization across neural language models},
  booktitle = {ICLR},
  year      = {2023}
}

@inproceedings{30,
  author    = {Kocher, P.},
  title     = {Timing attacks on implementations of {Diffie-Hellman}, {RSA}, {DSS}},
  booktitle = {CRYPTO},
  year      = {1996}
}

@inproceedings{31,
  author    = {Brumley, D. and Boneh, D.},
  title     = {Remote timing attacks are practical},
  booktitle = {USENIX Security},
  year      = {2003}
}

@inproceedings{32,
  author    = {Kocher, P. and others},
  title     = {Differential power analysis},
  booktitle = {CRYPTO},
  year      = {1999}
}

@inproceedings{33,
  author    = {Debenedetti, E. and others},
  title     = {Privacy side channels in machine learning systems},
  booktitle = {USENIX Security},
  year      = {2024}
}

@inproceedings{34,
  author    = {Staab, R. and others},
  title     = {Beyond memorization: Violating privacy via inference},
  booktitle = {ICLR},
  year      = {2024}
}

@inproceedings{35,
  author    = {Tram\`{e}r, F. and others},
  title     = {Truth serum: Poisoning machine learning models},
  booktitle = {CCS},
  year      = {2022}
}

@article{36,
  author    = {Nash, J.},
  title     = {The bargaining problem},
  journal   = {Econometrica},
  year      = {1950}
}

@article{37,
  author    = {Rubinstein, A.},
  title     = {Perfect equilibrium in a bargaining model},
  journal   = {Econometrica},
  year      = {1982}
}

@article{38,
  author    = {Myerson, R.},
  title     = {Mechanism design by an informed principal},
  journal   = {Econometrica},
  year      = {1983}
}

@book{39,
  author    = {Nisan, N. and others},
  title     = {Algorithmic Game Theory},
  publisher = {Cambridge University Press},
  year      = {2007}
}

@article{40,
  author    = {Fu, Y. and others},
  title     = {Improving language model negotiation with self-play},
  journal   = {arXiv preprint arXiv:2305.10142},
  year      = {2023}
}

@article{41,
  author    = {Abdelnabi, S. and others},
  title     = {{LLM}-powered multi-agent systems},
  journal   = {arXiv preprint arXiv:2308.04026},
  year      = {2023}
}

@article{61,
  author    = {Baarslag, T. and others},
  title     = {Learning about the opponent in automated bilateral negotiation},
  journal   = {Autonomous Agents and Multi-Agent Systems},
  year      = {2016}
}

@inproceedings{64,
  author    = {Bogetoft, P. and others},
  title     = {Secure multiparty computation goes live},
  booktitle = {Financial Cryptography},
  year      = {2009}
}

@inproceedings{70,
  author    = {Jagielski, M. and others},
  title     = {Auditing differentially private machine learning},
  booktitle = {NeurIPS},
  year      = {2020}
}

@inproceedings{71,
  author    = {Nasr, M. and others},
  title     = {Adversary instantiation: Lower bounds for differentially private machine learning},
  booktitle = {IEEE S\&P},
  year      = {2021}
}

@article{roy2026device,
  author    = {J. Roy and S. K. Singh},
  title     = {Device-Native Autonomous Agents for
               Privacy-Preserving Negotiations},
  journal   = {arXiv preprint arXiv:2601.00911},
  year      = {2026}
}

@article{juneja2025magpie,
  author    = {Gurusha Juneja and Jayanth Naga Sai Pasupulati
               and Alon Albalak and Wenyue Hua and
               William Yang Wang},
  title     = {{MAGPIE}: A Benchmark for Multi-{AG}ent
               Contextual Pr{I}vacy Evaluation},
  journal   = {arXiv preprint arXiv:2510.15186},
  year      = {2025}
}

\end{document}